%% file: main.tex
\documentclass[conference]{IEEEtran}
\usepackage{cite}
\usepackage{amsmath,amssymb,amsfonts}
\usepackage{amsmath, amsthm}
\usepackage{algorithmic}
\usepackage{graphicx}
\usepackage{textcomp}
\usepackage[dvipsnames]{xcolor}
\usepackage{dsfont}
\usepackage{siunitx}
\usepackage{comment}

\usepackage{pgfplots, pgfplotstable}
\pgfplotsset{compat=1.18}
\usetikzlibrary{intersections, perspective, shapes, positioning, calc, graphs, quotes}
\input{figs/tikz/tikz_definitions}

\begin{document}
\bstctlcite{IEEEexample:BSTcontrol}

\newcommand{\red}{\color{red}}
\newtheorem{proposition}{Proposition}
\newtheorem{definition}{Definition}
\newtheorem{lemma}{Lemma}

\title{Column Generation for the Optimization of Switching in Repeaterless Quantum Networks}

\author{
  Álvaro Troyano Olivas\IEEEauthorrefmark{1}\IEEEauthorrefmark{3}, Andrés Agustí Casado\IEEEauthorrefmark{2}, Hans H. Brunner\IEEEauthorrefmark{3},\\
  Chi-Hang Fred Fung\IEEEauthorrefmark{3}, Momtchil Peev\IEEEauthorrefmark{3}, Laura Ortiz\IEEEauthorrefmark{1}, and Vicente Martin\IEEEauthorrefmark{1}\\
  \IEEEauthorrefmark{1}\textit{Center for Computational Simulation, Universidad Politécnica de Madrid, Madrid, Spain}\\
  \IEEEauthorrefmark{2}\textit{Departamento de Física Teórica, Universidad Autónoma de Madrid, Madrid, Spain}\\
  \IEEEauthorrefmark{3}\textit{Munich Research Center, Huawei Technologies Duesseldorf GmbH, Munich, Germany}\\
  \texttt{alvaro.troyano.olivas@upm.es}
}

\maketitle

\begin{abstract}
Efficient resource allocation and optical switching promise high key rates, network adaptability, and cost reduction in repeaterless quantum communication networks.
However, identifying optimal switching configurations remains a significant challenge due to the combinatorial complexity.

We introduce a novel graph formulation to model the physical and logical structure of repeaterless quantum networks, enabling the systematic optimization of switching strategies.
The problem is posed as a linear program and solved using a column generation approach.
This method enables scalable computation despite the exponential number of possible network configurations.
Our results not only provide a formal foundation but also a practical algorithm for the optimization of switching.
Empirical tests confirm the solver's scalability with network size, demonstrating the framework's effectiveness and laying the groundwork for future optimization of quantum network control.
\end{abstract}

\begin{IEEEkeywords}
Repeaterless quantum communication; QKD; QKD networks; switched QKD; optimization; column generation
\end{IEEEkeywords}

\section{Introduction}

Quantum communication networks are rapidly evolving from laboratory experiments into operational infrastructures that enable unprecedented capabilities,
e.g., secure information transfer through quantum key distribution \cite{qkd}.
Recent testbeds such as MadQCI \cite{MadQCI:_a_heterogeneous_and_scalable_SDN-QKD_network_deployed_in_production_facilities} have demonstrated the feasibility of deploying these networks with flexible management of quantum and classical resources. 

A key element for the performance of such quantum networks is the ability to dynamically switch among key generation paths or configurations to optimize key generation rates, enhance network resilience, and ensure service continuity \cite{Demonstration_of_software_defined_network_services_utilizing_quantum_key_distribution_fully_integrated_with_standard_telecommunication_network, Quantum_key_distribution_based_on_selective_post-processing_in_passive_optical_networks, Demonstration_of_a_switched_CV-QKD_network}.
Architectures like such with quantum reconfigurable optical add-drop multiplexers (q-ROADMs) \cite{q-ROADM} are paving the way for scalable, flexible, and cost-efficient repeaterless quantum networks.
However, determining optimal switching strategies in these dynamic systems remains a significant and unresolved challenge.

In this contribution, we address this gap by developing a mathematical framework based on graph theory.
This framework models the physical structure of a repeaterless quantum network and determines how and when to perform switching of quantum links.
It also allows for the enumeration of all \textit{network configurations}
(i.e. each particular way to allocate the network resources), an unsolved problem, to the best of our knowledge.

Within this framework, we pose the switching optimization as a \textit{linear program} \cite{integer_program}.
To ensure scalability, we devise a solution method based on \textit{column generation}, which efficiently handles the problem's inherent complexity.
This approach contrasts with prior works \cite{Optimizing_key_consumption_in_switched_QKD_networks, Relayed_vs._Switched_QKD:_A_Comparison_in_Non-Uniform_Ring_Networks} that pose the problem as a computationally hard \textit{mixed integer linear program} \cite{integer_program}
and lack a formal modeling foundation.

Our analysis shows that the number of solver iterations required by the proposed algorithm scales roughly quadratically with the size of the network (this is, number of transmitters, receivers, switches and edges).
This quadratic growth is unaffected when the number of switches in the network is fixed.
Linear scaling is achieved when either the number of transmitters or receivers is fixed.
However, each iteration solves an integer subproblem, which is worst case exponential time.

The paper is structured as follows:
Section~\ref{sec:raw} states the problem and defines the mathematical framework.
Section~\ref{sec:linear_model} details the linear optimization model.
Section~\ref{sec:empirical_complexity} presents our scaling analysis,
and Section~\ref{sec:conclusions} concludes with directions for future research.



\section{Problem Setup}
\label{sec:raw}

The problem we address in this work is defined as follows:
Given a reconfigurable (\textit{switched}), repeaterless quantum communication network
consisting of transmitters and receivers with the capability to form various feasible pairs and known average key generation rates,
the objective is to determine the optimal distribution of network configurations over a specified time period.
In this section we establish the mathematical foundation required to formalize this optimization.

\begin{definition}[Raw Network Graph]
    A raw network graph $G = (V, E)$ is a directed graph with $E = \{(u, v): u, v \in V\}$ the set of \textit{arcs} (directed edges) and
    $V = T \cup R \cup S$ consisting of 3 types of nodes (which we will refer to as vertices onward):
    \text{transmitters} ($T$), \text{receivers} ($R$), and \text{switches} ($S$).
    We assume no ingoing edges exist for a transmitter and no outgoing edges for a receiver.
\end{definition}

\begin{figure}
    \centering
    \input{figs/tikz/example}
    \caption{%
      An example of a raw network graph,
      with transmitters shown as empty round vertices, receivers as filled round vertices and switches as square vertices.
      \label{fig:example_raw}}
\end{figure}
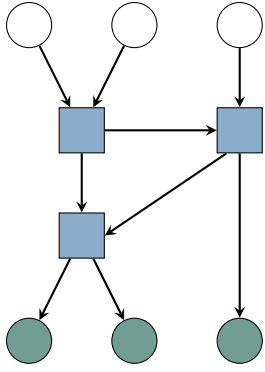

This framework can be extended to directed multigraphs, but without loss of generality and to simplify the notation we will assume the underlying graph is a digraph. For an example of a \textit{raw network graph}, the reader is referred to Fig.~\ref{fig:example_raw}.

\begin{definition}[TR-Path]
    A \text{TR-path} $p$ in a raw network graph $G = (V, E)$, from transmitter $t \in T$ to receiver $r \in R$, is a set of connected edges
    \begin{align}
        p = \{ (t, s_1), (s_1, s_2) \dots (s_n, r) \},
    \end{align}
    where $e_1 = (t, s_1) \in E$, $e_i = (s_{i - 1}, s_i) \in E$ and $e_{n+1} = (s_{n}, r) \in E$.
    
    A TR-path $p$ is a simple path if all vertices in $p$ are distinct. 
    Following path $p$, each vertex will only be visited at most one time, i.e., it is a loop-free path.
\end{definition}

In this work we will assume that the desirable (in the sense of useful for the optimization) paths are simple TR-paths, and hence, the word path will be used to refer to a simple TR-path.
For an example of a simple path, the reader is referred to Fig.~\ref{fig:path}.

\begin{figure}
    \centering
    \input{figs/tikz/simple_path}
    \caption{%
      A simple path in an example raw network graph, shown as orange edges.%
      \label{fig:path}}
\end{figure}
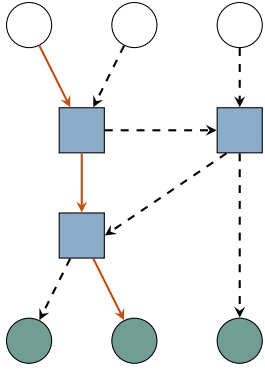

\begin{definition}[Link and Set of all Links]
  A transmitter-receiver pair is realizable if a TR-path exists that connects them.
  Such a realizable pair might also be called link.
  The set $L$ is the set of all realizable transmitter-receiver pairs in the raw network graph.
\end{definition}

In the context of the present problem, we need to define network configurations, since the optimization concerns them directly.

\begin{definition}[Network Configuration]
\label{def:network_config}
    Let $G = (T \cup S \cup R, E)$ be a raw network graph. A network configuration is a set $\tilde M$ of pairwise edge disjoint simple TR-paths in $G$ such that
    \begin{enumerate}
        \item For every transmitter $t \in T$, at most a single path $p \in \tilde M$ contains an edge with $t$.
        \item For every receiver $r \in R$, at most a single path $p \in \tilde{M}$ contains an edge with $r$.
    \end{enumerate}
\end{definition}

Note that a configuration is exactly a flow in the raw network graph $G$ setting each edge in $G$ to have unit capacity, and considering nodes in $T$ to be sources and nodes in $R$ to be sinks.

\begin{proposition}[Network Configuration]
\label{prop:network_configs}
    A \textit{network configuration} in the raw network graph $G = (T \cup R \cup S, E)$ is a mapping 
    \begin{align}
        M: E \rightarrow \{0, 1\}
    \end{align}
    such that, 
    \begin{enumerate}
        \item For every $t \in T$:
        \begin{equation}
            \label{eq:net_conf1}
            -1 \leq \sum_{(u, t) \in E} M(u, t) - \sum_{(t, v) \in E} M(t, v) \leq 0
        \end{equation}

        \item For every $r \in R$:
        \begin{equation}
            \label{eq:net_conf2}
            0 \leq \sum_{(u, r) \in E} M(u, r) - \sum_{(r, v) \in E} M(r, v) \leq 1
        \end{equation}

        \item For every $s \in S$:
            \begin{equation}
                \label{eq:net_conf3}
                \sum_{(u, s) \in E} M(u, s) - \sum_{(s, v) \in E} M(s, v) = 0
            \end{equation}
    \end{enumerate}    
    This corresponds to a unit capacity flow in graph $G$.
\end{proposition}

A proof for proposition \ref{prop:network_configs} can be found in appendix \ref{app:proof}.


\subsection{Weighted Key Capacity as Path Weights}

Given a raw network graph $G = (V, E)$, we will allow every edge $e \in E$ to have a channel attenuation $\alpha_e$ in decibel.
The secret key capacity of a path can then be calculated with the rate-loss scaling of repeaterless quantum communication defined in \cite{Fundamental_limits_of_repeaterless_quantum_communications}.
We will use a weighted secret key capacity as weight $\omega_{P(M, l)}$ of the path for link $l$ in configuration $M$:
\begin{align}
  \alpha_{P(M, l)} &= \sum_{e \in P(M, l)} \alpha_e,\\
  \omega_{P(M, l)} &= \begin{cases}
      - c_l \log_2\left(1 - 10^{-\frac{\alpha_{P(M, l)}}{10}}\right), \text{ if } P(M, l) \neq \emptyset \\
      0, \text{ if } P(M, l) = \emptyset.
  \end{cases}
\end{align}
where $P(M, l)$ is a function that gives the path in configuration $M$ that connects the transmitter-receiver pair of link $l$.
$c_l\ge0$ is used as an additional weight to reflect the priority of a link \cite{Optimizing_key_consumption_in_switched_QKD_networks}.
In this contribution, the attenuation of the edges, the secret key capacity, and the link weights are assumed to be known constants.

\subsection{Link Buffers and Time Evolution}

Each link in the network ($l \in L$) has a finite buffer to store generated keys.
This buffering allows keys to be consumed on demand, even when a link is not part of the active configuration.
To prevent any buffer from depletion, the system must switch between different network configurations.
This design effectively decouples the immediate demand for keys from the slower, time-consuming process of network reconfiguration. 

Key forwarding can establish end-to-end keys between any pair of nodes in trusted-node networks.
Finding key-forwarding routes is a separate problem outside the scope of this work.
We treat any forwarded key as a consumption demand on the individual links along its route.

For our analysis, we normalize all parameters relative to the total duration of a configuration sequence.
We assume the buffers are sufficiently large and initially filled such that a calculated non-negative average key gain guarantees a non-negative gain over the evolution and prevents buffers from depletion.
This allows us to analyze system behavior without tracking precise buffer levels over time.

Although key generation pauses during switching, we consider the time required to switch to be negligible. 
This assumption is justified because sufficiently large buffers allow very low switching frequencies (e.g, no more than once per hour). 

\section{Linear Model}
\label{sec:linear_model}

We formulate the optimization problem using a linear program based on  the raw network graph model.
Over a given finite time period, we seek to determine the fractions of time $\lambda_M$ that the network should operate in each configuration $M \in \mathfrak{M}$ to maximize a chosen objective, where $\mathfrak{M}$ denotes the set of all possible network configurations.
These fractions need to sum to one: $\sum_{M \in \mathfrak{M}} \lambda_M = 1$,
where $\sum_{M \in \mathfrak{M}} (\bullet)$ denotes summing over all possible values of $M$ (so all possible configurations).

We define the total amount of weighted key generation rate in link $l \in L$ as 
\begin{equation}
    b_l = \sum_{M \in \{M \in \mathfrak{M}: l \in M\}} \lambda_M \omega_{P(M, l)},
\end{equation}
where $l \in  M$ denotes link $l$ being realized by configuration $M$.

Among various possible objective functions, we adopt the max-min optimization strategy to maximize the minimum weighted key generation rate across all links:
\begin{equation}
    \max_{\lambda_M \geq 0,\; \sum_{M \in \mathfrak{M}} \lambda_M = 1}\quad \min_{l \in L} \sum_{M \in \{M \in \mathfrak{M}: l \in M\}} \lambda_M \omega_{P(M, l)}.
\end{equation}

We can linearize this objective by introducing an auxiliary variable $k$, which serves as a lower bound for the weighted per-link generation rate.
The resulting linear program is:
\begin{align}
    \label{prob:linear}\max_{\lambda_M \geq 0} \quad & k \\
    \text{s.t.} \quad
    & \sum_{M \in \{M \in \mathfrak{M}: l \in M\}} \omega_{P(M, l)} \lambda_M \geq k \quad \forall \, l \in L, \nonumber \\
    & \sum_{M} \lambda_M = 1.  \nonumber
\end{align}
In regard to the physical meaning, $k$ can be viewed as a measure of the maximum weighted per-link consumption rate supported by the network for a given strategy \cite{Optimizing_key_consumption_in_switched_QKD_networks}.

\subsection{Complexity of Column Enumeration}

The reader should note that this problem scales poorly, due to the exponential number of configurations ($\lambda_M$ variables).
Intuitively, this is because a configuration is essentially a set of edge disjoint paths.
The number of configurations is lower bounded by the number of paths themselves, because each path alone forms a network configurations, respectively.
The number of simple paths in a graph is, in the worst case, exponential in the number of edges ($\mathcal{O}(\exp(|E|))$).
In the worst case, the number of configurations is the number of subsets of the set of all simple paths, the size of its power set ($2^{\beta\exp(|E|)}$ with $\beta\exp(|E|)$ simple paths).
This number is an upper bound in networks with active edge-disjointness requirements, 
but the exponential relationship still establishes that the search space is too large for a brute-force solution approach.

Acknowledging that some configurations supersede other configurations might be a good way to reduce the search space.
However, this does not necessarily reduce the computational complexity of identifying the relevant configurations, since still, an exponential amount could exist.


\subsection{Solving the Problem}

The exponential number of variables in problem~\eqref{prob:linear} makes a full enumeration of all constraint matrix columns computationally infeasible.
We therefore employ \textit{column generation} \cite{integer_program}, an algorithm designed to tractably solve such large-scale linear programs.

The method operates by iteratively solving a \textit{restricted master problem} (RMP), which is the original problem posed on only a subset of variables or columns.
We define this restricted set of columns (configurations) as $\mathcal{M} \subset \mathfrak{M}$.
In each iteration, a \textit{pricing problem} is solved to identify a new column outside of $\mathcal{M}$ to be added to the RMP.
The selection of the new column is based on the \emph{reduced cost}, the benefit of adding this variable to the current pseudo optimal basis.

Column generation is an exact algorithm, guaranteeing convergence to the optimal solution.
However, it provides no a priori bound on the number of iterations it requires to converge to this optimum.

The RMP is given by
\begin{align}
    \label{prob:rmp}\max_{\lambda_M \geq 0} \quad &k \\
    \text{s.t.} \quad
    & \label{cons:gen_rate} \sum_{M \in \{M \in \mathcal{M}: l \in M\}} \omega_{P(M, l)} \lambda_M \geq k \quad \forall \, l \in L, \\
    & \label{cons:convexity} \sum_{M \in \mathcal{M}} \lambda_M = 1. 
\end{align}

To leverage the column generation framework, we need the dual of the original RMP.
By associating dual variables $\mu_l$, $l \in L$, to constraints~\eqref{cons:gen_rate}
and $\gamma$ to the convexity constraint~\eqref{cons:convexity},
the dual program can be found as 
\begin{align}
    \min_{\mu_l \geq 0} \quad &\gamma \\
    \text{s.t.} \quad
    & \label{cons:config} \sum_{l \in M} \omega_{P(M, l)} \mu_l \leq \gamma \quad \forall \, M \in \mathcal{M}, \\
    & \sum_{l \in L} \mu_l \geq 1. 
\end{align}

To solve problem~\eqref{prob:rmp} with column generation, we need a \textit{pricing problem} that at each iteration gives the column (or new configuration, not in the RMP) with the most positive reduced cost.
For this, we want to find the configuration that maximally violates constraint~\eqref{cons:config}.
In essence, we would like to solve
\begin{align}
    \max_{M} \quad& \sum_{l \in M} \omega_{P(M, l)} \mu_l \\
    \text{s.t.} \quad
    &\label{cons:vio} \sum_{l \in M} \omega_{P(M, l)} \mu_l >  \gamma.
\end{align}

Constraint~\eqref{cons:vio} is in place because adding a column to the RMP with negative or zero reduced cost would not change the optimal solution, and hence, stall our iterative column generation.
If the pricing problem is unfeasible, we will stop as we have found the optimum.

\subsubsection{Pricing Problem Derivation}

A possible formulation for the pricing problem is the following one.
Since we can think of a configuration as an edge disjoint path set, we can encode a configuration using some binary variables $y_p \in \{0, 1\}$ where $p$ is a path connecting the transmitter-receiver pair $(t,r)$ of link $l$.
Its value is 1 if the path is part of the configuration, 0 otherwise.
However, we need to enforce several constraints to encode a compliant network configuration;
let $P$ be the set of all paths, $P^{(t,r)}$ the set of all paths connecting $t$ to $r$, and $\mu_p = \mu_l$, $\forall p \in P^{(t,r)}$: 
\begin{itemize}
    \item Maximum transmitter usage: \\
    A single transmitter $t \in T$ should only be used once.
    This can be enforced as
    \begin{equation}
        \sum_{r \in R} \sum_{p \in P^{(t,r)}} y_p \leq 1, \quad \forall \, t \in T.
    \end{equation}

    \item Maximum receiver usage: \\
    A completely analogous argument follows for receivers,
    \begin{equation}
        \sum_{t \in T} \sum_{p \in P^{(t,r)}} y_p \leq 1, \quad \forall \, r \in R.
    \end{equation}

    \item Maximum fiber (edge) usage:
    Equivalently, we can only use each fiber channel (directed edge) in the raw network graph at most once,
    \begin{equation}
        \sum_{p \in \left\{p \in P: e \in p\right\}} y_p \leq 1, \quad \forall \, e \in E.
    \end{equation}
\end{itemize}

Hence, the full problem can be formulated as
\begin{align}
    \label{path_pricing_milp} \max_{y_p \in \{0, 1\}} \quad & \sum_{p \in P} \omega_p \mu_{p} y_p \\
    \text{s.t.} \quad
    & \sum_{r \in R} \sum_{p \in P^{(t,r)}} y_p \leq 1, \quad \forall \, t \in T, \notag\\
    & \sum_{t \in T} \sum_{p \in P^{(t,r)}} y_p \leq 1, \quad \forall \, r \in R, \notag\\
    & \sum_{p \in \left\{p \in P: e \in p\right\}} y_p \leq 1, \quad \forall \, e \in E. \notag
\end{align}

\section{Empirical Complexity Analysis}
\label{sec:empirical_complexity}

In this section, we present an empirical analysis of the optimization problem posed in Equation~\eqref{prob:rmp}.
Our methodology involves randomly generating raw network graphs, solving the optimization for each, and recording the number of iterations required by the column generation algorithm.
We analyze the results under three distinct scenarios: full network growth, growth with a fixed number of receivers, and growth with a fixed number of switches.

Fig.~\ref{fig:total_growth} displays the average iteration count as the entire network scales (i.e., the number of transmitters, receivers and switches all increase).
The relationship exhibits non-linear growth.
To understand this behavior, we consider the other two scenarios.

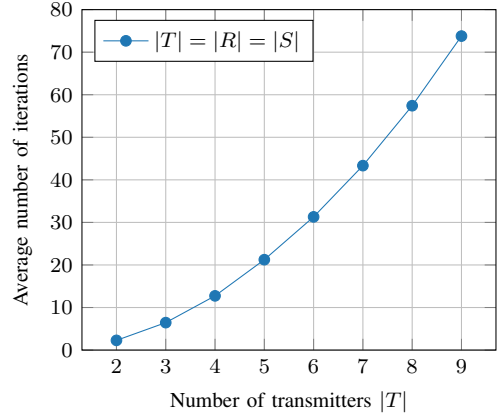
\begin{figure}
    \centering
    \input{figs/tikz/total_growth}
    \caption{%
      Mean number of iterations required to solve the optimization for random raw network graphs as a function of network size.
      The number of transmitters equals the number of receivers and switches.
      Each data point represents the average over 15 graph instances.%
      \label{fig:total_growth}}
\end{figure}

When the number of receivers is held constant (Fig.~\ref{fig:fixed_receivers}), the number of iterations scales linearly with the number of transmitters and switches.
This linear trend suggests that the non-linearity in the full-growth case is not inherent to the algorithm's core mechanics.
We hypothesize that it stems from the combinatorial increase in the number of possible transmitter-receiver pairs, which scales quadratically as $|T|\cdot|R|$.

\begin{figure}
    \centering
    \input{figs/tikz/fixed_receivers}
    \caption{%
      Mean number of iterations required to solve the optimization for random raw network graphs as a function of network size,
      while the number of receivers is held fixed at 5.
      Each data point represents the average over 15 graph instances.%
      \label{fig:fixed_receivers}}
\end{figure}
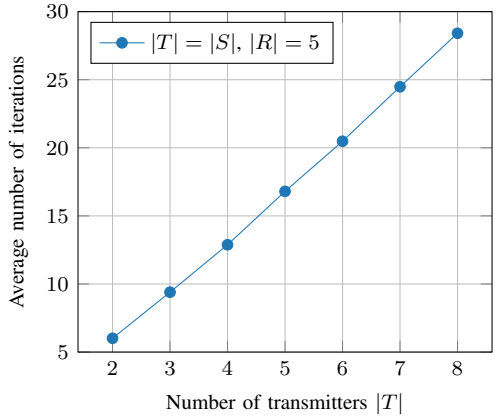

This hypothesis is further supported by the results in Fig.~\ref{fig:fixed_switches},
which shows non-linear growth when the number of switches is constant but the number of transmitters and receivers increases.
As we previously mentioned, given the linear growth when the number of receivers is constant,
this non-linear growth could be explained by the quadratic increase in the possible number of transmitter-receiver pairs.

\begin{figure}
    \centering
    \input{figs/tikz/fixed_switches}
    \caption{%
      Mean number of iterations required to solve the optimization for random raw network graphs as a function of network size,
      while the number of switches is held fixed at 5.
      Each data point represents the average over 15 graph instances.%
      \label{fig:fixed_switches}}
\end{figure}
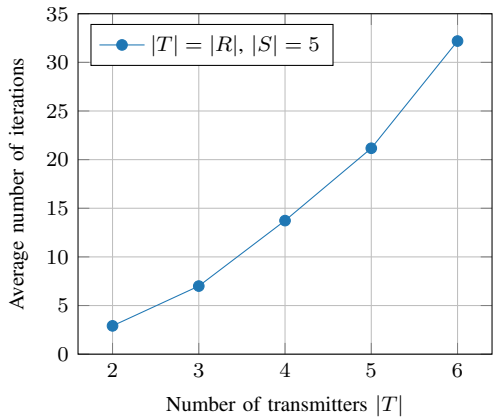

Although the number of column-generation iterations is well-behaved, each iteration requires solving the pricing problem given in Equation~\eqref{path_pricing_milp}.
This is a constrained binary optimization problem that is, in the worst case, exponential to solve in both time and computational steps.
Efficient exact methods and fast heuristics for this subproblem exist, but their consideration falls beyond the scope of this contribution.


\section{Conclusions}
\label{sec:conclusions}

In this work, we developed a mathematical optimization framework to characterize the switching configurations of a repeaterless quantum network.
Our results demonstrate that the column generation approach is a powerful method for optimizing these networks,
effectively solving a linear program with an exponential number of variables. 
While the pricing problem inside of the column generation is computationally expensive,
solving this problem is substantially more tractable than a full enumeration of all possible network configurations.


Future work should address the exploration of heuristic approaches for solving the pricing problem.
A fast approximation could dramatically accelerate convergence, making the approach viable for industrial protocols.
Additionally, time-dependent, statistical, and stochastic consumption and generation models should be investigated as well as the actual filling of the link buffers.
A more realistic model of switching time between configurations should be integrated.
Moreover, extending the mathematical framework to accommodate advanced protocols,
such as multi-party key distribution \cite{supporting_bob1,supporting_bob2},
represents a significant opportunity for future research.

\section{Acknowledgements}

The authors would like to thank projects EU Horizon Europe project ``Quantum Secure Networks Partnership" (QSNP), grant 101114043; the project ``MADQuantum-CM", funded by Comunidad de Madrid (Programa de Acciones Complementarias) and by the Recovery, Transformation and Resilience Plan—Funded by the European Union—NextGenerationEU (PRTRC17.I1) and the project ``CAM Programa TEC-2024/COM-84 QUITEMAD-CM".

\begin{appendices}

\section{Proof for Proposition 1}
\label{app:proof}

We prove that a mapping $M: E \rightarrow \{0, 1\}$ that verifies equations \eqref{eq:net_conf1}, \eqref{eq:net_conf2} and \eqref{eq:net_conf3} is a network configuration
if and only if the set $\tilde{M} = \{ \{(t, s_1), (s_1, s_2), \dots (s_k, r): M(t, s_1) = M(s_{i-1}, s_i) = M(s_k, r) = 1\}, \forall \, t \in T, r \in R\}$ is congruent with definition \ref{def:network_config}.

\begin{proof}
    Let $G = (T \cup R \cup S, E)$ be a raw network graph.

    \vspace{\baselineskip}

    $(\implies)$ Let $M: E \rightarrow \{0, 1\}$ be a mapping in $G$ such that it verifies equations \eqref{eq:net_conf1}, \eqref{eq:net_conf2} and \eqref{eq:net_conf3}.

    Let $t \in T$.
    Since no ingoing edges exist for a transmitter, $\sum_{(u, t) \in E} M(u, t) = 0$, Equation~\eqref{eq:net_conf1} simplifies to
    \begin{equation}
        -1 \leq - \sum_{(t, v) \in E} M(t, v) \leq 0.
    \end{equation}
    So at most a single $(t, v)$ edge can exist in $\tilde M$ (for a particular $t \in T$, and some other node $v \in S \cup R$).

    Furthermore, let $r \in R$.
    Again, since no outgoing edges exist for a receiver, $\sum_{(r, v) \in E} M(r, v) = 0$, Equation~\eqref{eq:net_conf2} simplifies to
    \begin{equation}
        0 \leq \sum_{(u, r) \in E} M(u, r) \leq 1.
    \end{equation}
    Hence at most a single $(u, r)$ can exist in $\tilde M$ (for a particular $r \in R$ and some node $u \in T \cup S$).
    
    Finally, Equation~\eqref{eq:net_conf3} 
    ensures that if a path reaches a switch on a dedicated edge, it must leave that same switch on a dedicated edge.
    Only transmitters can be sources and only receivers can be sinks of a flow in $G$.
    Every flow starts with a dedicated edge and, therefore, can only follow dedicated edges.

    This forms a set of edge disjoint simple paths, such that a transmitter can only pertain in a single one of them (equivalently for receivers).

    \vspace{\baselineskip}

    $(\impliedby)$ Let $\tilde M$ be a network configuration.
    Construct the mapping $M: E \rightarrow \{0, 1\}$ where $M(e) = 1 \iff e \in p$ for some $p \in \tilde M$.

    Then, for every $t \in T$, at most one path has $t$ as its source and no path has $t$ as its sink, so Equation~\eqref{eq:net_conf1}
    holds.
    Similarly, for every $r \in R$, at most one path contains $r$ as its target and no path contains $r$ as source, so Equation~\eqref{eq:net_conf2}
    holds.
    Finally, for every $s$ that some path contains, Equation~\eqref{eq:net_conf3} holds
    because a path is a set of dedicated and connected edges (and so, if a node in $S$ is reached, it must be left).
\end{proof}
\end{appendices}

\bibliographystyle{IEEEtran}
\bibliography{bibliography}

\end{document}

%% file: figs/tikz/tikz_definitions.tex
\definecolor{myblue}{RGB}{140, 173, 202}
\definecolor{mygreen}{RGB}{113, 157, 146}
\definecolor{myotherblue}{RGB}{31, 119, 180}

\tikzset{
  vertex/.style={
    draw,
    circle,
    minimum height=0.6cm,
    minimum width=0.6cm,
  },
  tx/.style={
    vertex,
  },
  rx/.style={
    vertex,
    fill=mygreen,
  },
  sw/.style={
    vertex,
    rectangle,
    fill=myblue,
  },
  connect/.style={
    thick,
    -stealth,
  },
  hinted/.style={
    connect,
    dashed,
  },
  simple/.style={
    connect,
    Bittersweet,
  },
}

%% file: figs/tikz/example.tex
\begin{tikzpicture}[scale=0.7]
  \node[tx] at (0, 0) (tx1) {};
  \node[tx] at (2, 0) (tx2) {};
  \node[tx] at (4, 0) (tx3) {};
  \node[sw] at (1,-2) (sw1) {};
  \node[sw] at (4,-2) (sw2) {};
  \node[sw] at (1,-4) (sw3) {};
  \node[rx] at (0,-6) (rx1) {};
  \node[rx] at (2,-6) (rx2) {};
  \node[rx] at (4,-6) (rx3) {};

  \draw[connect] (tx1) -- (sw1);
  \draw[connect] (tx2) -- (sw1);
  \draw[connect] (tx3) -- (sw2);

  \draw[connect] (sw1) -- (sw2);
  \draw[connect] (sw1) -- (sw3);
  \draw[connect] (sw2.240) -- (sw3.east);

  \draw[connect] (sw3) -- (rx1);
  \draw[connect] (sw3) -- (rx2);
  \draw[connect] (sw2) -- (rx3);
\end{tikzpicture}

%% file: figs/tikz/simple_path.tex
\begin{tikzpicture}[scale=0.7]
  \node[tx] at (0, 0) (tx1) {};
  \node[tx] at (2, 0) (tx2) {};
  \node[tx] at (4, 0) (tx3) {};
  \node[sw] at (1,-2) (sw1) {};
  \node[sw] at (4,-2) (sw2) {};
  \node[sw] at (1,-4) (sw3) {};
  \node[rx] at (0,-6) (rx1) {};
  \node[rx] at (2,-6) (rx2) {};
  \node[rx] at (4,-6) (rx3) {};

  \draw[simple] (tx1) -- (sw1);
  \draw[hinted] (tx2) -- (sw1);
  \draw[hinted] (tx3) -- (sw2);

  \draw[hinted] (sw1) -- (sw2);
  \draw[simple] (sw1) -- (sw3);
  \draw[hinted] (sw2.240) -- (sw3.east);

  \draw[hinted] (sw3) -- (rx1);
  \draw[simple] (sw3) -- (rx2);
  \draw[hinted] (sw2) -- (rx3);
\end{tikzpicture}

%% file: figs/tikz/total_growth.tex
\begin{tikzpicture}[font=\footnotesize]
  \begin{axis}[
      width=0.8\linewidth,
      xtick distance = 1,
      ytick distance = 10,
      grid = major,
      legend pos=north west,
      xlabel = {Number of transmitters $|T|$},
      ylabel = Average number of iterations,
      ymin=0, ymax=80,
    ]
    \addplot[myotherblue, mark=*] coordinates {
      (2,   2.28)
      (3,   6.45)
      (4,  12.74)
      (5,  21.24)
      (6,  31.30)
      (7,  43.34)
      (8,  57.42)
      (9,  73.78)
    };
    \addlegendentry{$|T|=|R|=|S|$}
  \end{axis}
\end{tikzpicture}

%% file: figs/tikz/fixed_receivers.tex
\begin{tikzpicture}[font=\footnotesize]
  \begin{axis}[
      width=0.8\linewidth,
      xtick distance = 1,
      ytick distance = 5,
      grid = major,
      legend pos = north west,
      xlabel = {Number of transmitters $|T|$},
      ylabel = Average number of iterations,
      ymin=5, ymax=30,
    ]
    \addplot[myotherblue, mark=*] coordinates {
      (2, 6.01)
      (3, 9.40)
      (4, 12.88)
      (5, 16.80)
      (6, 20.48)
      (7, 24.48)
      (8, 28.41)
    };
    \addlegendentry{$|T|=|S|$, $|R|=5$}
  \end{axis}
\end{tikzpicture}

%% file: figs/tikz/fixed_switches.tex
\begin{tikzpicture}[font=\footnotesize]
  \begin{axis}[
      width=0.8\linewidth,
      xtick distance = 1,
      ytick distance = 5,
      grid = major,
      legend pos = north west,
      xlabel = {Number of transmitters $|T|$},
      ylabel = Average number of iterations,
      ymin=0, ymax=35
    ]
    \addplot[myotherblue, mark=*] coordinates {
      (2, 2.91)
      (3, 7.00)
      (4, 13.73)
      (5, 21.17)
      (6, 32.19)
    };
    \addlegendentry{$|T|=|R|$, $|S|=5$}
  \end{axis}
\end{tikzpicture}